\documentclass[a4paper,11pt]{article}
\usepackage[margin=1in]{geometry}

\usepackage{amssymb}
\usepackage{amsthm}
\usepackage{mathtools}
\usepackage{tikz}
\usepackage{booktabs}
\usepackage{multirow}
\usepackage{subcaption}

\usepackage[
    numbers,
    square,
]{natbib}

\usepackage[
    colorlinks,
    linkcolor=black,
    citecolor=black,
    urlcolor=black,
]{hyperref}

\usepackage[
    capitalise,
]{cleveref}

\usetikzlibrary{
    fit,
    positioning,
}

\newtheorem{theorem}{Theorem}[section]
\newtheorem{proposition}[theorem]{Proposition}
\newtheorem{corollary}[theorem]{Corollary}
\newtheorem{lemma}[theorem]{Lemma}
\theoremstyle{definition}
\newtheorem{definition}[theorem]{Definition}

\newcommand{\defas}{\ensuremath{\coloneqq}}
\newcommand{\asdef}{\ensuremath{\eqqcolon}}

\title{
    Complexity of equilibria in binary public goods games on undirected graphs%
}
\author{
    Max Klimm\textsuperscript{1}
    \and
    Maximilian J. Stahlberg\textsuperscript{1}
}
\date{
    \medskip
	\small \textsuperscript{1}Technische Universit\"{a}t Berlin, Germany \\
	\small \texttt{$\{$klimm, stahlberg$\}$@math.tu-berlin.de}
}

\begin{document}

\maketitle

\begin{abstract}
    We study the complexity of computing equilibria in binary public goods
    games on undirected graphs. In such a game, players correspond to vertices
    in a graph and face a binary choice of performing an action, or not. Each
    player's decision depends only on the number of neighbors in the graph who
    perform the action and is encoded by a per-player binary pattern. We show
    that games with decreasing patterns (where players only want to act up to a
    threshold number of adjacent players doing so) always have a pure Nash
    equilibrium and that one is reached from any starting profile by following
    a polynomially bounded sequence of best responses. For non-monotonic
    patterns of the form $10^k10^*$ (where players want to act alone or
    alongside $k + 1$ neighbors), we show that it is $\mathsf{NP}$-hard to
    decide whether a pure Nash equilibrium exists. We further investigate a
    generalization of the model that permits ties of varying strength: an edge
    with integral weight $w$ behaves as $w$ parallel edges. While, in this
    model, a pure Nash equilibrium still exists for decreasing patters, we show
    that the task of computing one is $\mathsf{PLS}$-complete.
\end{abstract}

\section{Introduction}

\emph{Public goods} are resources that can be freely accessed and
simultaneously used by many individuals. In the physical world, they comprise
abundant natural resources like sunlight and breathable air alongside
artificial goods such as cultural heritage, public art, or early warning
systems. Public goods have become ubiquitous in the information age, when data
can be reproduced at a negligible cost, yet remains valuable to its users.
Examples from this domain are open source software, public databases, radio
transmissions, and the diffusion of scientific knowledge through open channels.
However, not all public goods are universal: a population warning system
profits a region, herd immunity to a virus is enjoyed on the basis of human
contact, and a scientific manuscript may be legible only to a group of peers.
For these scenarios, the possibility of access may be represented by a graph,
where a vertex can only enjoy goods that are provided by itself or by one of
its neighbors. In the classical setting of \citet{Bramoulle07}, each vertex may
produce any amount of a fixed good and experiences utility based on the total
amount of the good that is available in the closed neighborhood, minus a cost
associated with local production. In the \emph{public goods game}, vertices aim
to configure their own production to maximize the resulting utility. Whether
such strategic decisions allow for a stable arrangement and, if so, how one can
be found and proposed to the actors are natural questions for policy makers.
\citeauthor{Bramoulle07} show for strictly concave benefits and linear costs
the existence of \emph{specialized equilibria}, in which the production of each
vertex is either zero or a fixed, positive amount. This motivates the study of
a binary variant of the game, wherein vertices have only two options: to be
\emph{active} and produce (one unit of) the good, or to remain
inactive~\cite{Gilboa22,Kempe20,Maiti22,Yang20,Yu20}. In general, costs and
utility functions can be defined in a way that allows vertices to be
indifferent to this choice. \citet{Yang20} use this to show that it is
\(\mathsf{NP}\)-hard to decide whether an equilibrium exists, already for a
monotone benefit function common to all agents. In response, \citet{Maiti22}
introduce---and \citet{Gilboa22} shift the focus towards---restricted games in
which every vertex~\(i\) always has a strict preference with respect to its
options. As this preference depends only on the number of active neighbors, we
may encode it by an infinite, indexed-from-zero binary sequence~\(T^i\), called
a \emph{best response pattern}~\cite{Gilboa22}, where \(T^i_\ell\) is the
activity response of~\(i\) to exactly~\(\ell\) active neighbors. When there is
a single pattern common to each vertex, this is called the \emph{fully
homogeneous} case~\cite{Yu20}, otherwise we call the game inhomogeneous.

\subsection{Our results}

We investigate the computational hardness of deciding whether a game admits a
pure Nash equilibrium, for two natural classes of patterns for which this
question was unresolved~\cite{Gilboa22}.

Decreasing patterns are those that start with a positive number of ones and are
zero thereafter. We show that for this class, even the inhomogeneous binary
public goods game is equivalent to a congestion game. We use this equivalence to
show that best-response dynamics converge to a pure Nash equilibrium in
polynomial time. For an extension of the game in which edges have an integral
weight, however, we show that finding an equilibrium is $\mathsf{PLS}$-complete.

We call a picky pattern one where \(T_{\ell} = 1\) if and only if \(\ell \in
\{0, k + 1\}\), for some \(k \geq 1\). Here, vertices wish to act only alone or
alongside \(k + 1\) neighbors. For example, when \(k = 1\) for every vertex,
then the Graph induced by the active vertices of an equilibrium consists of
singletons and cycles. We show that deciding the existence of such an
equilibrium is an $\mathsf{NP}$-complete problem, even in the fully homogeneous
case. A notable corollary is the following: While deciding the infinite
alternating sequence \(T = (1, 0, 1, 0, \ldots)\) is in
$\mathsf{P}$~\cite{Papadimitriou21}, any truncation of this sequence in which
it is zero after alternating at least two times corresponds to a hard problem.

\begin{table}
    \centering
    \begin{tabular}{c l c}
        \toprule
        Complexity & Pattern class & Reference \\
        \midrule
        \(O(1)\)
        & \(1^+0^+\) & \cref{monotone}\(^\dagger\) \\
        \cmidrule{1-1}
        \(O(\operatorname{poly}(n))\)
        & \({(10)}^*\) & \cite[Lemma 3.3]{Papadimitriou21} \\
        \cmidrule{1-1}
        \multirow{4}{*}{\textsf{NP}-complete}
        & \(10^+10^*\) & \cref{picky-pattern} \\
        & \({(10)}^+10^*\) & \cref{truncated-alternating}\(^\ddagger\) \\
        & \(11.^*0.^*10^*\) & \cite[Theorem 5]{Gilboa22} \\
        & \(10^+11.^*0^+\) & \cite[Theorem 6]{Gilboa22} \\
        \bottomrule
    \end{tabular}
    \caption{Complexity of deciding the existence of equilibria in the
        homogeneous binary public goods games on undirected graphs, for varying
        classes of non-trivial (not starting with zero) best response patterns.
        The pattern class notation is defined in \cref{pattern-notation}.
        Informally, \(x^*\)~and~\(x^+\) denote, respectively, zero-or-more and
        one-or-more repetitions of pattern~\(x\) and the dot is a placeholder
        for either \(0\)~or~\(1\) (repeated dots may take distinct values).
        \(^\dagger\)The special case of \(10^*\) has been shown
        in~\cite{Bramoulle07}, that of \(110^*\) in~\cite{Gilboa22}.
        \(^\ddagger\)Uses \cite[Theorem 7]{Gilboa22}.}%
    \label{table}
\end{table}

\Cref{table} shows a summary of our results for homogeneous patterns alongside
previous results.

\subsection{Related work}

Let $G = (V,E)$ be an undirected graph and denote the \emph{open neighborhood}
of vertex $i \in V$ by $N(i) \defas \{ j \in V : \{i,j\} \in E\}$.
\citet{Bramoulle07} study a model where each vertex is a player whose strategy
is to pick a production effort $s_i \in S_i \defas \mathbb{R}_{\geq 0}$. For a
strategy vector $s = (s_1,\dots,s_n)$, the utility of each player is given by
$u_i(s) \defas b(s_i + \sum_{j \in N(i)} s_j) - c s_i$, where $b \colon
\mathbb{R}_{\geq 0} \to \mathbb{R}_{\geq 0}$ is a benefit function with $b(0) =
0$, $b' > 0$ and $b'' < 0$, and where $c > 0$ denotes the cost of production.
The authors show the existence of a particular kind of Nash equilibrium, $s$
with $s_i \in \{0, {(b')}^{-1}(c)\}$, that they call a \emph{specialized
equilibrium}. They translate this result also to the more general case where
$b_i$~and~$c_i$ depend on the player~$i$. \citet{Pandit18} discuss the
efficiency of specialized equilibria and give conditions under which they
attain the maximum social welfare, maximum production effort, and minimum cost
achieved by a general Nash equilibrium.

\citet{Yu20} study a binary variant of the game with $S_i \defas \{0,1\}$ and
where they allow for arbitrary non-decreasing and player-specific~$b_i$, and
for player-specific~$c_i$. They aim to show that in this setting, deciding the
existence of a pure Nash equilibrium is an $\mathsf{NP}$-complete problem.
\citet{Yang20} point out a flaw in the proof and provide a corrected version.
They further extend this result to homogeneous games where all functions~$b_i$
and all parameters~$c_i$ are equal. As noted by \citet{Gilboa22}, the hardness
result of \citeauthor{Yang20} hinges on the fact that some players are
indifferent to their binary decisions, i.e., $b(k+1) - b(k) = c$ for some value
of~$k$. \citeauthor{Gilboa22} exclude this possibility by studying
\emph{best-response patterns}. These are binary sequences that explicitly give
the best response of a player \(i\) based on the number of neighbors \(j \in
N(i)\) with \(s_j = 1\). The authors settle the complexity for a number of
patterns, including some that start with~\(0\), for which they consider
\emph{non-trivial equilibria}. They also show that hard patterns that start
with~\(1\) remain hard when they are prefixed with~\((1, 0)\). See \cref{table}
for previous results on patterns starting with~\(1\). In independent and
concurrent work, \citet{Gilboa23} gives an alternative proof of
\cref{picky-pattern} and settles the complexity of all remaining finite (i.e.,
\(T \in .^*0^+\)) and non-monotone patterns, which complements our results on
monotone patterns to a full characterization of the finite and homogeneous
case. \citet{Maiti22} provide a parametrized complexity perspective wherein
they study natural graph parameters such as the maximum degree and the
treewidth.

\citet{LopezPintado13} initiates the study of public goods games on directed
graphs. The author analyzes a restricted binary model where each agent~$i$
prefers to play $s_i = 1$ if and only if none of their in-neighbors~$j$ plays
$s_j = 1$. \citet{Papadimitriou21} perform a complexity analysis of the binary
variant on directed graphs for more general utility functions. They investigate
in particular the setting where all players share a utility function that is
monotone in the number of in-neighbors~$j$ who play $s_j = 1$. This allows for
a pattern-wise analysis akin to that of \citet{Gilboa22}. Here the authors
provide a complete characterization by showing that only the trivial all-zero
and all-ones patterns as well as the infinite alternating sequence starting
with~\(1\) are polynomial-time decidable. They further provide a
$\mathsf{PPAD}$-hardness result for the computation of approximate mixed Nash
equilibria.

\citet{Kempe20} are interested in modifying the graph in such a way that one
element of a given set of strategy vectors is a pure Nash equilibrium of the
binary game. The motivation behind this is to enforce a socially preferable
equilibrium through the intervention of a policy maker. Finally,
\citet{Galeotti10} study a variant of the game where players have only partial
knowledge of the network structure.

\section{Preliminaries}

For \(n \in \mathbb{Z}_{\geq 1}\) we write \([n]\) short for \(\{1, \ldots,
n\}\). For a set \(A\), we write \(\mathbf{1}_A\) for the indicator function of
\(A\), that is \(\mathbf{1}_A(x) \defas 1\), if \(x \in A\), and
\(\mathbf{1}_A(x) \defas 0\), if \(x \not\in A\). For an undirected and simple
graph \(G = (V, E)\), we write \(\Delta(G)\) for its maximum vertex degree,
\(\deg(v)\) for the degree of \(v \in V\), \(G[X]\) for the subgraph induced by
\(X \subseteq V\), \(N_G(v)\) for the open neighborhood of \(v \in V\), and
\(N_G(X) \defas \{v \in V \setminus X \mid \exists x \in X \colon \{v, x\} \in
E\}\) for the open neighborhood of \(X \subseteq V\). We drop suffixes when the
graph in question is clear.

When \(s \in {\{0, 1\}}^n\) with \(n = |V|\) denotes a strategy profile of the
binary public goods game played on \(G = (V, E)\) with best-response pattern
\(T\), then we say that a vertex \(v \in V\) is \emph{active (in \(s\))}, if
\(s_v = 1\), and \emph{inactive}, if \(s_v = 0\). We further write \(\deg_s(v)
\defas |\{u \in N_G(v) \mid s_u = 1\}|\) and call this the \emph{active degree}
of \(v\). Finally, we say that a vertex \(v \in V\) is \emph{best-responding
(in~\(s\))} if \(s_v = T_{\deg_s(v)}\).

\subsection{Response pattern notation}\label{pattern-notation}

We use regular expressions (REs) over the alphabet \(A \defas \{0,1\}\) to
describe classes of best response patterns. We first define their syntax: Every
\(c \in A\) and the dot \(.\) are RE.\@ If \(e\)~and~\(f\) are REs, then so are
the concatenation \(ef\), the \(k\)-fold concatenation of \(e\) with itself,
\({(e)}^k\), the zero-or-more operation~\({(e)}^*\), and the one-or-more
operation~\({(e)}^+\). Parentheses may be omitted: superscripts then take
precedence over concatenation, for example \(10^* = 1{(0)}^* \neq {(10)}^*\).

The language \(L(e)\) generated by an RE \(e\) is the following: For the dot we
have \(L(.) \defas \{(a) \mid a \in A\}\). For \(c \in A\), it is \(L(c) \defas
\{(c)\}\). For REs \(e\)~and~\(f\), it is \(L(ef) \defas L(e) \times L(f)\),
\(L(e^k) \defas {L(e)}^k\), \(L(e^+) \defas \bigcup_{i=1}^\infty {L(e)}^i\) and
\(L(e^*) \defas L(e^+) \cup \{()\}\), where \(()\) is the empty sequence over
\(A\) and where \((c) = c\) is a sequence of length one.

\emph{Extensible} REs end in \({(e)}^*\) or \({(e)}^+\) for some sub-expression
\(e\). For the set of best response patterns \(P(e)\) generated by an
extensible RE~\(e\), we have \({(x)}_{i=1}^\infty \in P(e)\) if and only if
there is a threshold~\(N\) such that \({(x)}_{i=1}^n \in L(e)\) for all \(n
\geq N\). For example, \(P(1^*0^+) = P(1^*00^*)\) contains all infinite
sequences that start with a nonnegative, finite number of ones and are zero
thereafter, while \(P(1^*0^*)\) contains in addition the infinite sequence of
ones. In the following, we write \(T \in e\) short for \(T \in P(e)\). We also
write \(T = e\) instead of \(T \in e\) to signify that \(|P(e)| = 1\).

\subsection{Best-response games, strategic games, and consistent
    representations}

In the following we make explicit the relation between the binary public goods
game that is specified by the players' best-response behaviors and the
classical definition where players associate a utility value with their two
choices. This work adopts the former representation as a \emph{best-response
game} that captures precisely the \emph{strict}~\cite{Maiti22} public goods
games on graphs. To still leverage established notions of game equivalence that
are based on utility functions, we introduce the notion of a \emph{consistent
representation} of a best-response game as a \emph{strategic game} in which
players experience utility. We start the introduction of these concepts with
the class of best-response games:

\begin{definition}[Best-response game]
    A best-response game (BRG) is described by a triple \(\left(n,
    {(S_i)}_{i=1}^n, {(\beta_i)}_{i=1}^n\right)\) where \(n\) is the number of
    players, \(S_i\) is the strategy set of player \(i \in [n]\), and \(\beta_i
    \colon \prod_{j=1}^{i-1} S_j \times \prod_{j=i+1}^{n} S_j \to S_i\) is a
    function such that \(\beta_i(s_{-i})\) denotes the best response of player
    \(i\) given that the other players play \(s_{-i}\).

    For a strategy profile \(s \in \prod_{j=1}^n S_j\) and \(i \in [n]\), we
    write \(s_{-i}\) short for \(\prod_{j=1}^{i-1} \{s_j\} \times
    \prod_{j=i+1}^{n} \{s_j\}\) and we call \(s\) a \emph{pure Nash
    equilibrium} (PNE) if \(\beta_i(s_{-i}) = s_i\) for all \(i \in [n]\).
\end{definition}

It is straightforward to formalize the binary public goods game as a BRG.\@

\begin{definition}[Binary public goods game]\label{pgg}
    For an undirected simple graph \(G = (V, E)\) with \(n\) vertices \(V =
    [n]\) and for best-response patterns \(T^i =
    {(\tau^i_\ell)}_{\ell=0}^\infty\) for all \(i \in V\), we call the
    best-response game \(\left(n, {(\{0, 1\})}_{i=1}^n,
    {(\beta_i)}_{i=1}^n\right)\) with \(\beta_i(s_{-i}) \defas
    T^i_{x_i(s_{-i})}\) and \(x_i(s_{-i}) \defas \sum_{j \in N_G(i)} {s_{j}}\)
    a (binary) public goods game (PGG). We also refer to just \(G\) and the
    \(T^i\) as a PGG.\@

    If \(T^i = T\) for all \(i \in V\) and for some pattern \(T\), then we
    refer to the associated PGG as a \emph{homogeneous PGG with pattern \(T\)}.
\end{definition}

While in a best-response game a player's response is defined uniquely, in
strategic games players choose a strategy among all that maximize a utility
function:

\begin{definition}[Strategic game]
    A strategic game is a triple \(\left(n, {(S_i)}_{i=1}^n,
    {(u_i)}_{i=1}^n\right)\) where \(n\) is the number of players, \(S_i\) is
    the strategy set of player \(i \in [n]\), and \(u_i \colon \prod_{j=1}^n
    S_j \to \mathbb{Q}\) is a function such that \(u_i(s)\) denotes the utility
    experienced by player \(i \in [n]\) when the strategy profile \(s \in
    \prod_{j=1}^n S_j\) is played.

    For a strategy profile \(s \in \prod_{j=1}^n S_j\), \(i \in [n]\), and \(a
    \in S_i\), we write \((s_{-i},a)\) short for \(\prod_{j=1}^{i-1} \{s_j\}
    \times \{a\} \times \prod_{j=i+1}^{n} \{s_j\}\). We call a strategy profile
    \(s \in \prod_{j=1}^n S_j\) with \(u_i(s) \geq u_i(s_{-i},a)\) for all \(i
    \in [n]\) and all \(a \in S_i\) a \emph{pure Nash equilibrium} (PNE).
\end{definition}

With a suited utility function, we can represent any BRG as a strategic game
with equivalent best-response dynamics (cf.~\cite{Maiti22}):

\begin{definition}[Consistent representation]\label{conrep}
    Let \(A = \left(n, {(S_i)}_{i=1}^n, {(\beta_i)}_{i=1}^n\right)\) be a
    BRG.\@ We call a strategic game \(B = \left(n, {(S_i)}_{i=1}^n,
    {(u_i)}_{i=1}^n\right)\) a \emph{consistent representation} of \(A\) if for
    all \(s \in \prod_{j=1}^n S_j\) and all \(i \in [n]\), it is
    \(u_i(s_{-i},a) > u_i(s_{-i},a')\) for all \(a' \in S_i \setminus \{a\}\)
    where \(a \defas \beta_i(s_{-i})\) is the best response of \(i\) in \(A\).
\end{definition}

Consistent representations are closely related to the notion of a \emph{weak
    isomorphism} between strategic games~\cite{Gabarro11}. In particular, they
also preserve the structure of pure Nash equilibria:

\begin{lemma}
    Let \(A = \left(n, {(S_i)}_{i=1}^n, {(\beta_i)}_{i=1}^n\right)\) be a BRG
    and \(B = \left(n, {(S_i)}_{i=1}^n, {(u_i)}_{i=1}^n\right)\) a consistent
    representation of \(A\). Then, \(s \in \prod_{j=1}^n S_j\) is a PNE of
    \(A\) if and only if \(s\) is a PNE of \(B\).
\end{lemma}

\begin{proof}
    Let first \(s\) be a PNE of \(A\) and consider a player \(i \in [n]\) of
    \(A\) and their best response \(a \defas \beta_i(s_{-i})\). Since \(A\) is
    a PNE, it is \(a = s_i\). As further B is a consistent representation,
    \begin{align*}
        u_i(s) = u_i(s_{-i},s_i) = u_i(s_{-i},a) > u_i(s_{-i},a')
    \end{align*}
    for all \(a' \in S_i \setminus \{a\}\). Thus, \(u_i(s) \geq
    u_i(s_{-i},a'')\) for all \(a'' \in S_i\), so \(s\) is a PNE of \(B\).

    Let next \(s\) be a PNE of \(B\) and consider a player \(i \in [n]\) of
    \(B\) playing \(a \defas s_i\). Assume towards a contradiction that \(s_i
    \neq \beta_i(s_{-i}) \asdef b\). Since \(B\) is a PNE, \(u_i(s_{-i},a) =
    u_i(s) \geq u_i(s_{-i},a')\) for all \(a' \in S_i\). From \(b \in S_i\), it
    follows in particular that \(u_i(s_{-i},a) \geq u_i(s_{-i},b)\). As \(B\)
    is a consistent representation and \(a \neq b\), \(u_i(s_{-i},b) >
    u_i(s_{-i},a)\), a contradiction. Hence, \(s_i = \beta_i(s_{-i})\), so
    \(s\) is a PNE of \(A\).
\end{proof}

We are further interested in the dynamics that may lead to a PNE:\@

\begin{definition}[Better-response sequence]
    For a best-response game \(\left(n, {(S_i)}_{i=1}^n,
    {(\beta_i)}_{i=1}^n\right)\), a better-response sequence is a sequence of
    strategy profiles \({(s^i)}_{i=1}^N\) such that for all \(i \in [N - 1]\),
    it is \(s^{i+1} = (s^i_{-j},a)\) for some \(j \in [n]\) and \(a \in S_j\)
    with \(\beta_j(s^i_{-j}) = a\) and \(a \neq s^i_j\).

    For a strategic game \(\left(n, {(S_i)}_{i=1}^n, {(u_i)}_{i=1}^n\right)\),
    a better-response sequence is a sequence of strategy profiles
    \({(s^i)}_{i=1}^N\) such that for all \(i \in [N - 1]\), it is \(s^{i+1} =
    (s^i_{-j},a)\) for some \(j \in [n]\) and \(a \in S_j\) such that
    \(u_j(s^i_{-j},a) > u_j(s^i)\).
\end{definition}

It is immediate from \cref{conrep} that consistent representations preserve
better-response sequences:

\begin{lemma}
    If \(A\) is a BRG and \(B\) a consistent representation of \(A\), then any
    better-response sequence in \(A\) is also a better-response sequence in
    \(B\).
\end{lemma}

Note that the converse does not hold: While the better response of a player in a
best-response game is unique, there may be multiple unilateral strategy changes
that improve the utility of the deviating player in a strategic game.

\subsection{Game isomorphisms}

We will make use of the following notion of equivalence of strategic games:

\begin{definition}[Strong isomorphism \citep{Gabarro11}]%
    \label{isomorphism}
    Two strategic games \(A = \left(n, {(S_i)}_{i=1}^n,
    {(u_i)}_{i=1}^n\right)\) and \(B = \left(n, {(S'_i)}_{i=1}^n,
    {(u'_i)}_{i=1}^n\right)\) are called \emph{(strongly) isomorphic} when
    there are bijections \(\pi \colon [n] \to [n]\) and \(\phi_i \colon S_i \to
    S'_{\pi(i)}\) for all \(i \in [n]\) such that for all players \(i \in [n]\)
    (of \(A\)) and strategy profiles \(s \in \prod_{j=1}^n S_j\), it is
    \(u_i(s) = u'_{\pi(i)}\left(\phi(s)\right)\) where \(\phi(s) \defas {\left(
    \phi_{\pi^{-1}(j)}(s_{\pi^{-1}(j)}) \right)}_{j=1}^n\).
\end{definition}

Note that for \(\pi = \operatorname{id}_{[n]}\), it is \(\phi(s) = {\left(
\phi_j(s_j) \right)}_{j=1}^n\). Informally, an isomorphism describes a
one-to-one mapping between the players and strategy profiles of two games that
preserves the players' utility and, by extension, the players' strategy
preferences and the PNE structure of both games. In particular:

\begin{lemma}\label{isomorphisms-preserve-pne}
    If two games \(A\) and \(B\) are isomorphic and \(A\) admits a PNE, then so
    does \(B\).
\end{lemma}

\begin{proof}
    If \(A = \left(n, {(S_i)}_{i=1}^n, {(u_i)}_{i=1}^n\right)\) has a PNE \(s
    \in \prod_{i=1}^n S_i\), then \(u_i(t) \leq u_i(s)\) for all \(i \in [n]\)
    and \(t \in \prod_{j=1}^n S_j\). As \(A\) and \(B\) are isomorphic, there
    exist bijections \(\pi\) and \(\phi\) such that
    \begin{align*}
        u'_{\pi(i)}\left(\phi(t)\right)
        =
        u_i(t)
        \leq
        u_i(s)
        =
        u'_{\pi(i)}\left(\phi(s)\right)
    \end{align*}
    for all \(i\) and \(t\) as above. As both \(\pi\) and \(\phi\) are
    surjective, \(s'\) is a PNE of \(B\).
\end{proof}

\subsection{Encoding}

For the homogeneous PGG, we assume in line with previous work that the common
best response pattern is part of the problem definition and not of the problem
input. For the inhomogeneous variant, this approach is not adequate as each
player might have a distinct pattern, so here we assume that the patterns are
part of the input and give the encoding explicitly.

\section{Existence of equilibria for decreasing patterns}

Decreasing best-response patterns are those that begin with a positive number
of ones and are zero thereafter. They arise naturally in scenarios where public
goods have a constant cost to produce and provide diminishing returns (to the
producer and its neighbors alike), expressed through a concave utility
function. We show that for this family of patterns, formally \(P(1^+0^+)\), the
public goods game is equivalent to a \emph{congestion game}:

\begin{definition}[Congestion game]
    Let \(E\) be a set of goods equipped with a delay function \(d_e \colon
    \mathbb{Z}_{\geq 1} \to \mathbb{Q}\) for every \(e \in E\). Then, we call a
    strategic game \(\left(n, {(S_i)}_{i=1}^n, {(u_i)}_{i=1}^n\right)\) with
    \(S_i \subseteq E\) and \(u_i(s) = -\sum_{e \in s_i} d_e(x_s(e))\) with
    \(x_s(e) \defas \sum_{j=1}^n \mathbf{1}_{s_j}(e)\) for all \(i \in [n]\) a
    congestion game.
\end{definition}

More precisely, we show equivalence in the following sense:

\begin{proposition}\label{isomorphic}
    The binary public goods game on undirected graphs with best-response
    patterns \(T^i \in 1^+0^+\), \(i \in V\), has a consistent representation
    that is isomorphic to a congestion game.
\end{proposition}

In the proof we construct a congestion game as follows: Every edge and every
vertex of the public goods game \(G = (V, E)\) is represented by a congestible
good in the congestion game \(C\). A player \(i \in V\) remaining inactive in
\(G\) is represented by the associated player \(i \in [n]\) choosing the
singleton strategy \(\{i\}\) in \(C\), which has a constant cost just below the
number of ones in the player's pattern, \(k_i\). The alternative strategy of
\(i\) in \(C\) corresponds to \(i\) being active in \(G\) and amounts to
selecting the set of all edges incident to \(i\) in \(G\). The cost of this
strategy in \(C\) is the number of edges that are used by both players who have
the edge in their strategy. The idea is that players will strictly prefer the
``edge'' strategy as long as less than \(k_i\) of its neighbors do so, and the
``vertex'' strategy otherwise.

\begin{proof}[Proof of \Cref{isomorphic}]
    Let \(G = (V, E)\) be an instance of the PGG with patterns \(T^i =
    1^{k_i}0^*\), \(k_i \geq 1\), for all \(i \in V = [n]\). Define the
    strategic game \(\Gamma \defas \left(n, {(S_i)}_{i=1}^n,
    {(u_i)}_{i=1}^n\right)\) with
    \begin{align}
        S_i    & \defas \{0, 1\}~\text{and} \notag \\
        u_i(s) & \defas
        \begin{cases}
            -(k_i - \frac{1}{2}), & \text{if}~s_i = 0, \\
            -\deg_s(i),           & \text{if}~s_i = 1, \\
        \end{cases} \label{eq_utility}
    \end{align}
    for all \(i \in V\).

    We first show that \(\Gamma\) is a consistent representation of \(G\). Let
    \(s \in {\{0,1\}}^n\) be a strategy profile and \(i \in V\) a player of
    \(G\). Let further \(a \defas \beta_i(s_{-i})\) be the best response of
    \(i\) and \(a' \defas 1 - a\) its unique alternative. We show that then
    \(u_i(s_{-i},a) > u_i(s_{-i},a')\). If \(a = 0\), then \(\deg_s(i) \geq
    k_i\), thus
    \begin{align*}
        u_i(s_{-i},a)
        = -\left(k_i - \frac{1}{2}\right)
        > -k_i
        \geq -\deg_s(i)
        = u_i(s_{-i},a').
    \end{align*}
    On the other hand, if \(a = 1\), then \(\deg_s(i) \leq k_i - 1\) and
    \begin{align*}
        u_i(s_{-i},a)
        = -\deg_s(i)
        \geq -(k_i - 1)
        > -\left(k_i - \frac{1}{2}\right)
        = u_i(s_{-i},a').
    \end{align*}

    Consider next the congestion game \(C = \left(n, {(S'_i)}_{i=1}^n,
    {(u'_i)}_{i=1}^n\right)\) defined by the set of goods \(V \cup E\) with
    delays \(d_v(x) \defas k_v - \frac{1}{2}\), for \(v \in V\), and \(d_e(x)
    \defas x - 1\), for \(e \in E\), and by the strategies \(S'_v \defas
    \{\{v\}, \{e \in E \mid v \in e\}\}\) for all \(v \in V\). The strategy
    \(\{v\}\) corresponds to \(v\) remaining inactive, which has no effect on
    adjacent vertices, while \(\{e \in E \mid v \in e\}\) corresponds to \(v\)
    being active, which makes being active more expensive for \(v\)'s neighbors
    by ``congesting'' incident edges.

    We show that \(C\) is isomorphic to \(\Gamma\) as witnessed by the
    bijections \(\pi \defas \operatorname{id}_{[n]}\) and \(\phi_v \colon S_v
    \to S_v'\) with \(\phi_v(0) = \{v\}\) and \(\phi_v(1) = \{e \in E \mid v
    \in e\}\) for all \(v \in V\). To this end, let \(v \in V\) be a player and
    \(s \in \prod_{i=1}^n S_i\) a strategy profile of \(\Gamma\). If \(s_v =
    0\), then \(\phi_v(s_v) = \{v\}\), and thus
    \begin{align*}
        u_v(s)
        = -\left(k_v - \frac{1}{2}\right)
        = - d_v(x_s(\{v\}))
        = u'_v(\phi(s)).
    \end{align*}
    If \(s_v = 1\), then \(\phi_v(s_v) = \{e \in E \mid v \in e\}\), so that
    \begin{align*}
        u_v(s)
        = ~& -\deg_s(v) \\
        = ~& -\sum_{\{v, i\} \in E} s_i \\
        = ~& -\sum_{\{v, i\} \in E}
            \Big| \Big\{
                j \in [n] \setminus \{v\}
                \mid
                s_j = 1
                \land
                \underbrace{\{v, i\} \in \{e \in E \mid j \in e\}%
                    }_{j = i~\text{as}~j \neq v}
            \Big\} \Big| \\
        = ~&~ \begin{aligned}[t]
                - &\sum_{\{v, i\} \in E}
                \left| \left\{
                    j \in [n]
                    \mid
                    s_j = 1 \land \{v, i\} \in \{e \in E \mid j \in e\}
                \right\} \right| \\
                + &\sum_{\{v, i\} \in E}
                \Big| \Big\{
                    j \in \{v\}
                    \mid
                    s_j = 1
                    \land
                    \underbrace{\{v, i\} \in \{e \in E \mid j \in e\}%
                        }_{\text{true as}~j = v}
                \Big\} \Big|
            \end{aligned} \\
        \intertext{by the definitions of \(u_v\) for \(s_v = 1\), and
            \(\deg_s\). It follows that}
        u_v(s)
        = ~& -\left(
                \sum_{\{v, i\} \in E}
                    \left|\{ j \in [n] \mid \{v, i\} \in \phi_j(s_j) \}\right|
                - \deg_G(v)
            \right) \\
        = ~& -\sum_{\{v, i\} \in E} \left(
            \sum_{j=1}^n \mathbf{1}_{\phi_j(s_j)}(\{v, i\}) - 1 \right) \\
        \intertext{by the definition of \(\phi_j\) for \(s_j = 1\). Further,}
        u_v(s)
        = ~& -\sum_{\{v, i\} \in E} d_{\{v, i\}}\left(
            \sum_{j=1}^n \mathbf{1}_{\phi_j(s_j)}(\{v, i\}) \right) \\
        = ~& -\sum_{e \in \{e \in E \mid v \in e\}} d_e\left(
            \sum_{i=1}^n \mathbf{1}_{{\phi(s)}_i}(e) \right) \\
        \intertext{follows from the definition of \(d_e\) for \(e \in E\).
            Finally,}
        u_v(s)
        = ~& -\sum_{e \in \phi_v(s_v)} d_e\left( x_{\phi(s)}(e) \right) \\
        = ~& u'_v(\phi(s)),
    \end{align*}
    which concludes the proof.
\end{proof}

This observation allows us to state the main theorem of this section, which
answers the first out of two open questions in~\cite{Gilboa22}:

\begin{theorem}\label{monotone}
    The binary public goods game on undirected graphs with best-response
    patterns \(T^i \in 1^+0^+\), \(i \in V\), always admits a PNE;\@ the
    associated decision problem is thus in $\mathsf{P}$.
\end{theorem}

\begin{proof}
    As every congestion game has a PNE~\cite{Rosenthal73}, the claim follows
    from \cref{isomorphisms-preserve-pne,isomorphic}.
\end{proof}

Additionally, the representation as a congestion game brings with it
convergence guarantees:

\begin{corollary}\label{quadratic-convergence}
    In the game of \cref{monotone}, best-response dynamics converge to a PNE
    after at most \(O(n^2)\) improving steps.
\end{corollary}

\begin{proof}
    Let \(k_{\max} \defas \max_{i \in V} k_i\). Any vertex can have at most
    \(\Delta(G)\) active neighbors, so the game for \(k_{\max} > \Delta(G) +
    1\) is equivalent to the one for \(k_{\max} = \Delta(G) + 1\). Let thus
    \(k_{\max} \leq \Delta(G) + 1 \leq n\), without loss of generality.

    The congestion game \(C\) in the proof of \cref{isomorphic} is an exact
    potential game~\cite{Monderer96} with potential function
    \begin{align*}
        \Phi(s)
        = \sum_{g \in V \cup E} \sum_{\ell = 1}^{x_s(g)} d_e(\ell)
        = \sum_{\{u, v\} \in E} s_u s_v
            + \sum_{v \in V} \left( k_v - \frac{1}{2} \right) s_v
        \leq |E| + k_{\max} |V|.
    \end{align*}
    From \(d_e(\ell) \geq 0\) for all \(\ell \geq 1\) and \(e \in V \cup E\),
    it follows that \(\Phi(s) \geq 0\) for all \(s \in \prod_{i \in [n]}
    S'_i\). On the other hand, it is \(\Phi(s) \leq |E| + k_{\max} |V| \leq 2
    n^2\). As \(2 \Phi\) is integral by construction, these bounds imply that
    any decreasing sequence \(\Phi(s_1) > \ldots{} > \Phi(s_N)\) has length at
    most \(O(n^2)\). The claim follows as any better-response sequence in \(G\)
    corresponds to a better-response sequence in \(C\), which has decreasing
    potential.
\end{proof}

\section{Hardness of decreasing patterns with weighted edges}

We next investigate a natural extension of the public goods game, in which ties
can have varying importance:

\begin{definition}[Weighted binary public goods game]
    For an undirected simple graph \(G = (V, E)\) with \(n\) vertices \(V =
    [n]\), edge weights \(w_e \in \mathbb{Z}_{\geq 1}\) for all \(e \in E\),
    and best-response patterns \(T^i = {(\tau^i_\ell)}_{\ell=0}^\infty\) for
    all \(i \in V\), we call the best-response game \(\left(n, {(\{0,
    1\})}_{i=1}^n, {(\beta_i)}_{i=1}^n\right)\) with \(\beta_i(s_{-i}) \defas
    T^i_{x_i(s_{-i})}\) and \(x_i(s_{-i}) \defas \sum_{j \in N_G(i)}
    {w_{\{i,j\}} s_{j}}\) an (edge-)weighted binary public goods game.
\end{definition}

We first notice that a straightforward adaptation of the proof of
\cref{isomorphic} yields that also games with weighted edges are isomorphic to
a congestion game:

\begin{proposition}\label{weighted-isomorphic}
    The weighted binary public goods game on undirected graphs with
    best-response patterns \(T^i \in 1^+0^+\), \(i \in V\), has a consistent
    representation that is isomorphic to a congestion game.
\end{proposition}

\begin{proof}[Proof (Sketch)]
    We use the same construction as in the proof of \cref{isomorphic} except
    that the utilities in the definition of \(\Gamma\) (\cref{eq_utility}) are
    given by \(u_i(s) \defas - x_i(s_{-1})\) for \(s_i = 1\) and the delay
    functions are defined as $d_e(x) \defas w_e(x-1)$ for all $e \in E$.
\end{proof}

In \cref{quadratic-convergence}, we have shown that best-response dynamics
converge in polynomial time to a PNE for unweighted binary public goods games.
In contrast, we show that for weighted games, the computation of a PNE is
$\mathsf{PLS}$-complete. For the proof, we use a straightforward reduction from
the problem of computing a pure Nash equilibrium in a threshold game, a
particular kind of congestion game where each player has two strategies only.
(Not to be confused with another ``threshold game'' introduced
in~\cite{Papadimitriou21}, where players pick their strategy from the unit
interval.) The only non-trivial part in the reduction is the fact that in a
threshold game, a player may be indifferent between their two strategies, while
in a weighted binary public goods game this cannot occur. Intuitively, the
absence of indifference makes the computation of equilibria for public goods
games only harder. For the proof, we formalize this intuition.

\begin{theorem}
    Computing a pure Nash equilibrium of a weighted binary public goods game on
    an undirected graph with patterns $T^i \in 1^+0^+$ for all $i \in [n]$ is
    $\mathsf{PLS}$-complete when patterns are encoded through the number of
    leading ones as a binary number.
\end{theorem}

\begin{proof}
    We first sketch membership in $\mathsf{PLS}$: Feasible solutions of a
    weighted binary public goods game are strategy profiles, encoded as binary
    vectors of size \(n\). Their value may be defined as the potential value of
    the congestion game that is isomorphic to the game's consistent
    representation, which can be computed in polynomial time. As an initial
    solution one may choose an all-zero vector; the neighborhood of a solution
    is just the flip-neighborhood of size \(n\).

    Ackermann et al.~\cite[Theorem~4.1]{Ackermann08} have shown that the
    computation of a pure Nash equilibrium for a threshold game is
    $\mathsf{PLS}$-complete. A threshold game is a congestion game where the
    set of goods $E$ is partitioned into two disjoint sets $E_{\text{in}}$ and
    $E_{\text{out}}$. The set $E_{\text{out}} \defas \{e_i \mid i \in [n]\}$
    contains a good $e_i$ for every player~$i$. The set $E_{\text{in}}$
    contains a good $e_{i,j}$ for every unordered pair of players $\{i,j\}
    \subseteq [n]$ with $i \neq j$. The delay function of the goods $e_i \in
    E_{\text{out}}$ is constant, i.e., for every player~$i$ there is a constant
    $\theta_i \in \mathbb{R}_{> 0}$ such that $d_{e_i}(x) = \theta_i$ for all
    $x \in \mathbb{Z}_{\geq 1}$. As explained in
    \cite[Remark~4.2]{Ackermann08}, the proof of
    \cite[Theorem~4.1]{Ackermann08} only uses delay functions of the form
    \begin{align}
        \label{eq:pls-reduction-delay-functions}
        d_{e_{i,j}}(x) = a_{i,j} (x-1) \quad \text{ for all $\{i,j\}
            \subseteq [n]$ with $i \neq j$, where $a_{i,j} > 0$.}
    \end{align}
    A closer examination of the proof further reveals that the values $a_{i,j}$
    can in fact be chosen to be integer. Thus, the $\mathsf{PLS}$-completeness
    also holds for this special case. The set of strategies of each player~$i$
    is given by $S_i = \{s^{\text{out}}_i, s^{\text{in}}_i\}$ with
    $s_i^{\text{out}} = \{e_i\}$ and $s_i^{\text{in}} = \{e_{i,j} \mid j \in
    [n], j \neq i\}$.

    For the reduction, let an instance of a threshold game with the delay
    functions as in~\eqref{eq:pls-reduction-delay-functions} be given. We
    construct a corresponding instance of a weighted binary public goods game
    as follows. The game is played on a complete graph $G = (V,E)$ with edge
    weights $w_e = a_{i,j} \in \mathbb{Z}_{\geq 1}$ for all $e = \{i,j\} \in
    E$. We further set $T^i \defas 1^{\lfloor \theta_i \rfloor}0^*$ for all $i
    \in [n]$. Let $s$ be a pure Nash equilibrium of the thus defined weighted
    binary public goods game. We claim that $\bar{s}$ defined for all $i \in
    [n]$ as $\bar{s}_i = s_i^{\text{out}}$, if $s_i = 0$, and $\bar{s}_i =
    s_i^{\text{in}}$, if $s_i = 1$, is a pure Nash equilibrium of the threshold
    game. By the definition of~$\bar{s}$ and the fact that $s$ is a pure Nash
    equilibrium of the weighted binary public goods game, we have $\bar{s}_i =
    s_i^{\text{out}}$ whenever $x_i(s_{-i}) \geq \lfloor \theta_i \rfloor + 1$
    and $\bar{s}_i = s_i^{\text{in}}$ whenever $x_i(s_{-i}) \leq \lfloor
    \theta_i \rfloor$. A player~$i$ with $\bar{s}_i = s_i^{\text{out}}$ has
    utility $-\theta_i$. After a deviation to $s_i^{\text{in}}$, the utility
    would be $-x_i(s_{-i}) \leq -(\lfloor \theta_i \rfloor +1) \leq -\theta_i$,
    so that this deviation is not profitable. On the other hand, a player~$i$
    with $\bar{s}_i = s_i^{\text{in}}$ has utility $-x_i(s_{-i})$. After a
    deviation to $s_i^{\text{out}}$, the utility would be $-\theta_i \leq
    -\lfloor \theta_i \rfloor \leq -x_i(s_{-i})$, so that also this deviation
    is not profitable. This concludes the $\mathsf{PLS}$-reduction.
\end{proof}

\section{Hardness of the picky pattern}

In a picky pattern, players want to perform the action if either none or a
particular number of neighbors also does so, formally \(T = 10^k10^*\) with \(k
\geq 1\). We define a number of graph gadgets that are used in a
polynomial-time reduction from the \textsf{NP}-hard~\cite{Garey79}
\texttt{POSITIVE-1IN3-SAT} problem:

\begin{definition}
    An instance of the \texttt{POSITIVE-1IN3-SAT} problem is a collection of
    \(\ell\) clauses \(I \defas \left\{ \{l^i_1, l^i_2, l^i_3\} \mid i \in
    [\ell] \right\}\) with \(l^i_j \in X \cup \{\bot\}\) for all \((i, j) \in
    [\ell] \times [3]\), where \(X = \{\xi_1, \ldots, \xi_m\}\) is a set of
    boolean variables and where \(\bot\) denotes unconditional falsity. The
    problem is to decide whether there is a truth assignment to the variables
    in \(X\) satisfying
    \begin{align*}
        \Phi(X) \defas \bigwedge_{i \in [\ell]} \left(
            \left( l^i_1 \lor l^i_2 \lor l^i_3 \right)
            \land \lnot \left( l^i_1 \land l^i_2 \right)
            \land \lnot \left( l^i_2 \land l^i_3 \right)
            \land \lnot \left( l^i_3 \land l^i_1 \right)
        \right).
    \end{align*}
\end{definition}

A closely related problem was also used in a reduction to a binary public goods
game in~\cite{Gilboa22}.

In the reduction, the literals of a \texttt{POSITIVE-1IN3-SAT} instance (either
non-negated variables or falsity) are represented by \emph{literal vertices}
and truth assignments to the underlying boolean variables are encoded by these
vertices' strategies. The gadgets represent logical operators and connect the
literal vertices in such a way that the resulting graph admits a PNE if and
only if the \texttt{POSITIVE-1IN3-SAT} formula is satisfiable. To this end,
every gadget has a set of \emph{operand vertices} that are identified with a
subset of the literal vertices; no other gadget vertex has an edge leaving the
gadget. We refer to gadget vertices that are adjacent to the operator vertices
as \emph{membrane vertices}. Three out of four gadgets are constructed in such
a way that membrane vertices are inactive in any PNE on a graph containing the
gadget. We call this property \emph{safety} as it rules out potential side
effects when the gadget is added to a graph. In particular, this ensures that
adding a gadget to a graph that admits no PNE will not allow a PNE to exist in
the resulting graph:

\begin{lemma}\label{safety-lemma}
    Let \(G_\circ = (V_\circ, E_\circ)\) be a graph gadget with operand
    vertices \(X \subseteq V_\circ\) and membrane vertices \(M \defas
    N_{G_\circ}(X)\) such that for all graphs \(G' = (V', E')\) with \(V' \cap
    V_\circ = X\), for all PNE \(s\) on \((V' \cup V_\circ, E' \cup E_\circ)\),
    and for all \(m \in M\), it is \(s_m = 0\). Let further \(G = (V, E)\) be a
    fixed graph with \(V \cap V_\circ = X\) that admits no PNE.\@ Then, also
    \(H = (V \cup V_\circ, E \cup E_\circ)\) admits no PNE.\@
\end{lemma}

\begin{proof}
    Let \(G_\circ\), \(X\), \(M\), \(G\), and \(H\) as in the lemma and assume
    towards a contradiction that \(H\) admits a PNE \(s\). Consider the
    strategy profile \(t\) obtained by limiting \(s\) to vertices in \(G\). If
    \(\deg_t(v) = \deg_s(v)\) for all \(v \in V\), then \(s\) is a PNE on
    \(G\), so there is a \(v \in V\) with \(\deg_t(v) \neq \deg_s(v)\). If
    \(N_H(v) \subseteq V\), then \(\deg_t(v) = \deg_s(v)\) by construction of
    \(t\), so there is further a \(u \in N_H(v) \setminus V\) with \(s_u = 1\).
    Since \(u \not\in V\), \(\{u, v\} \not\in E\), so \(\{u, v\} \in E_\circ\),
    implying \(u, v \in V_\circ\). From \(v \in V\) it follows that \(v \in X\)
    and from \(u \in N_H(v)\) it follows that \(u \in N_{G_\circ}(v)\) and thus
    \(u \in N_{G_\circ}(X) = M\), contradicting \(s_u = 1\).
\end{proof}

\subsection{The \texttt{NEAR-OR} gadget}

\begin{figure}
    \centering
    \begin{subfigure}[b]{0.4\textwidth}
        \centering
        \begin{tikzpicture}
    \input{tikz_header}

    \node[membrane, label=above:{$w$}] (w) {};
    \node[private, below left=of w, label=left:{$y$}] (y) {};
    \node[private, below left=of y, label=below:{$y'$}] (y_) {};
    \node[private, below right=of w, label=right:{$z$}] (z) {};
    \node[private, below right=of z, label=below:{$z'$}] (z_) {};
    \node[private, below right=of y, label=below:{$q$}] (q) {};

    \begin{scope}[node distance=12mm]
        \node[public, above left=of w, label=above:{$x_1$}] (x_1) {};
        \node[public, above right=of w, label=above:{$x_\ell$}] (x_l) {};
    \end{scope}

    \node[group, fit=(x_1) (x_l), label=right:{$X$}] {};

    \draw (w) -- (y) -- (y_) -- (q) -- (y) -- (z) -- (q) -- (z_) -- (z) -- (w);
    \draw (x_1) -- (w);
    \draw (x_l) -- (w);
    \draw[leftout] (x_1) -- (x_l);
\end{tikzpicture}
        \caption{\texttt{NEAR-OR} gadget for \(k = 1\).}%
        \label{fig:near-or-gadget-1}
    \end{subfigure}
    \begin{subfigure}[b]{0.4\textwidth}
        \centering
        \begin{tikzpicture}
    \input{tikz_header}

    \node[membrane, label=above:{$w$}] (w) {};
    \node[private, below left=of w, label=left:{$y$}] (y) {};
    \node[private, below right=of w, label=right:{$z$}] (z) {};
    \node[private, below left=of y, label=below:{$y_1$}] (y_1) {};
    \node[private, right=of y_1, label=below:{$y_k$}] (y_k) {};
    \node[private, below right=of z, label=below:{$z_k$}] (z_k) {};
    \node[private, left=of z_k, label=below:{$z_1$}] (z_1) {};

    \begin{scope}[node distance=12mm]
        \node[public, above left=of w, label=above:{$x_1$}] (x_1) {};
        \node[public, above right=of w, label=above:{$x_\ell$}] (x_l) {};
    \end{scope}

    \node[group, fit=(x_1) (x_l), label=right:{$X$}] {};
    \node[group, fit=(y_1) (y_k), label=left:{$Y$}] {};
    \node[group, fit=(z_1) (z_k), label=right:{$Z$}] {};

    \draw (w) -- (y) -- (z) -- (w);
    \draw (y) -- (y_1); \draw (y) -- (y_k); \draw[leftout] (y_1) -- (y_k);
    \draw (z) -- (z_1); \draw (z) -- (z_k); \draw[leftout] (z_1) -- (z_k);
    \draw (x_1) -- (w);
    \draw (x_l) -- (w);
    \draw[leftout] (x_1) -- (x_l);
\end{tikzpicture}
        \caption{\texttt{NEAR-OR} gadget for \(k \geq 2\).}%
        \label{fig:near-or-gadget-k}
    \end{subfigure}
    \caption{%
        \texttt{NEAR-OR} gadget: In any PNE \(s\) on a graph that contains this
        gadget as a subgraph in such a way that non-black vertices are
        connected only to other gadget vertices, it is \(\sum_{i=1}^\ell
        s_{x_i} \not\in \{0, k + 1\}\). We refer to black vertices as
        \emph{operand vertices}, to gray vertices as \emph{membrane vertices},
        and to white vertices as \emph{internal vertices}.}%
    \label{fig:near-or-gadgets}
    \caption*{(a) A graph gadget comprising the \(3\)-sun graph \(S3\) and
        \(\ell\) additional vertices \(X\) attached to its top corner, which is
        labeled \(w\). The middle level of the \(S3\) is labeled \(y\) and
        \(z\), the lower level \(y'\), \(q\), and \(z'\). (b) A triangle with
        additional vertex groups \(X\), \(Y\), and \(Z\) attached to its
        corners.}
\end{figure}

The \texttt{NEAR-OR} gadget (\cref{fig:near-or-gadgets}) is used to ensure
(under the assumption that a PNE exists) that at least one literal in each
clause of a \texttt{POSITIVE-1IN3-SAT} instance evaluates to true. The ``near''
in its name stems from the fact that, when used to represent an \(\ell\)-ary
logical operator with \(\ell > k\), the gadget cannot distinguish between a
total of \(0\) or \(k + 1\) operands evaluating to true; in both cases the
gadget will prevent a PNE.\@ The gadget further appears as a building block in
other gadgets, where we make use of this property. We call the \texttt{NEAR-OR}
gadget for \(\ell = 1\) the \texttt{TRUE} gadget as it forces its single
operand vertex to be active in any PNE.\@

The following lemma implies that the \texttt{NEAR-OR} gadget forbids a PNE in
which none of its operand vertices are active.

\begin{lemma}\label{near-or-gadget_restrictive-0}
    The \texttt{NEAR-OR} gadget with operand vertices removed admits no PNE.\@
\end{lemma}

\begin{proof}[Proof for \(k = 1\)]
    Let \(G\) be the graph of \cref{fig:near-or-gadget-1} without the vertices
    in \(X\) and assume towards a contradiction that \(G\) admits a PNE \(s\).
    Assume further that \(\deg_s(v) \geq 4\) for some \(v \in V\). Then, \(v\)
    is a vertex with \(\deg(v) = 4\) and \(v\) is inactive in \(s\). Consider
    \(\{u, v\} \in E\) with \(\deg(u) = 2\). Then, \(u\) is active with
    \(\deg_s(u) = 1\) as \(N(u) = \{v, v'\} \in E\), contradicting that \(s\)
    is a PNE.\@ We have thus \(\deg_s(v) \leq 3\) for all \(v \in V\). In
    particular, \(v \in V\) is active if and only if \(\deg_s(v)\) is even. Let
    next \(A \subseteq V\) active and \(I = V \setminus A\) inactive. Since
    \(\sum_{a \in A} \deg_s(a) + \sum_{i \in I} \deg_s(i) = \sum_{a \in A}
    \deg(a)\) is even, also \(|I|\) and by extension \(|A| = 6 - |I|\) are
    even. The cases of \(|A| \in \{0, 6\}\) are easily ruled out, so either
    \(|A| = 2\) or \(|I| = 2\). If \(A = \{a_1, a_2\} \in E\), \(\deg_s(a_1) =
    1\) contradicts \(a_1\) active. If \(A = \{a_1, a_2\} \not\in E\),
    \(\operatorname{diam}(G) = 2\) implies an \(i \in I\) with \(A \subseteq
    N(i)\) so that \(\deg_s(i) = 2\) contradicts \(i\) inactive. If \(I =
    \{i_1, i_2\} \in E\), then \(G[A]\) is either the paw graph or the union of
    a \(P_3\) and a singleton. Both have a vertex of degree one, implying \(a
    \in A\) with \(\deg_s(a) = 1\), which contradicts \(a\) being active.
    Finally, if \(I = \{i_1, i_2\} \not\in E\), then \(\deg_s(i_1) = \deg(i_1)
    \in \{2, 4\}\). As we ruled out \(\deg_s(i_1) = 4\) earlier, this
    contradicts \(i_1\) inactive.
\end{proof}

\begin{proof}[Proof for \(k \geq 2\)]
    Let \(G\) be the graph of \cref{fig:near-or-gadget-k} without the vertices
    in \(X\) and assume towards a contradiction that \(G\) admits a PNE \(s\).
    If both \(y\) and \(z\) are inactive, then \(w\) and all vertices in \(Y\)
    and in \(Z\) have no active neighbors and are active. This contradicts
    \(y\) being inactive, as \(\deg_s(y) = |Y| + 1 = k + 1\). If both \(y\) and
    \(z\) are active, then \(\deg_s(v) \in \{1, 2\}\) for all \(v\) in \(\{w\}
    \cup Y \cup Z\), so all vertices other than \(y\) and \(z\) are inactive.
    Thus, \(\deg_s(y) = 1\), which contradicts \(y\) being active. If exactly
    one of \(y\) and \(z\) is active, say \(y\), then all vertices in \(Y\) are
    inactive and all in \(Z\) are active. If further \(w\) is inactive, then
    \(\deg_s(z) = k + 1\) contradicts \(z\) being inactive. If however \(w\) is
    active, then \(\deg_s(y) = 1\) contradicts \(y\) being active.
\end{proof}

The next lemma states that the \texttt{NEAR-OR} gadget also forbids a PNE in
which exactly \(k + 1\) of its operand vertices are active.

\begin{lemma}\label{near-or-gadget_restrictive-kp1}
    Let \(G_\lor = (V_\lor, E_\lor)\) an instance of the \texttt{NEAR-OR}
    gadget, \(G = (V, E)\) a graph with \(V \cap V_\lor = \{x_i \mid i \in
    [\ell]\}\), and \(H \defas (V \cup V_\lor, E \cup E_\lor)\). Then, \(H\)
    admits no PNE \(s\) with \(\sum_{i=1}^\ell s_{x_i} = k + 1\).
\end{lemma}

\begin{proof}[Proof for \(k = 1\)]
    Assume towards a contradiction that \(s\) is a PNE of \(H\) with
    \(\sum_{i=1}^\ell s_{x_i} = k + 1\). If \(w\) is active, then both \(y\)
    and \(z\) are inactive as otherwise \(\deg_s(w) > k + 1\). If further \(q\)
    is inactive, then both \(y'\) and \(z'\) are active as \(\deg_s(y') =
    \deg_s(z') = 0\). This contradicts \(q\) inactive, as \(\deg_s(q) = 2\), so
    \(q\) is active. If \(y'\) is inactive, then \(\deg_s(y) = 2\) contradicts
    \(y\) inactive. So \(y'\) must be active, contradicting \(\deg_s(y') = 1\).
    If \(w\) is inactive, then \(\deg_s(w) > k + 1\), so \(y\) or \(z\) or both
    are active. If both \(y\) and \(z\) are active, then exactly one of \(y'\)
    and \(q\) must be active, otherwise \(\deg_s(y) \in \{1, 3\}\). If \(q\) is
    active, \(\deg_s(y') = 2\) contradicts \(y'\) inactive. If \(y'\) is
    active, this contradicts \(\deg_s(y') = 1\). Thus, exactly one of \(y\) and
    \(z\) is active, say \(y\). Since \(\deg_s(y) \in \{0, 2\}\) with \(z\) and
    \(w\) both inactive, it follows that \(y'\) and \(q\) are either both
    active or both inactive. If both are active, this contradicts \(z\)
    inactive as then \(\deg_s(z) = 2\). If both are inactive, it is
    \(\deg_s(z') = 0\), so \(z'\) is active. This again implies \(\deg_s(z) =
    2\), contradicting \(z\) inactive.
\end{proof}

\begin{proof}[Proof for \(k \geq 2\)]
    Assume towards a contradiction that \(s\) is a PNE of \(H\) with
    \(\sum_{i=1}^\ell s_{x_i} = k + 1\). Analogous to the proof for \(k = 1\),
    we have that either \(w\) is active and both \(y\) and \(z\) are inactive,
    or that \(w\) is inactive and at least one of \(y\) and \(z\) is active. In
    the former case, all vertices in \(Y\) have no active neighbors and are
    active, contradicting \(y\) inactive as \(\deg_s(y) = |Y| + 1 = k + 1\). In
    the latter case, if both \(y\) and \(z\) are active, then all vertices in
    \(Y\) are inactive, contradicting \(y\) active as \(\deg_s(y) = 1\). If
    only \(z\) is active, then all vertices in \(Y\) are active, contradicting
    \(y\) inactive as again \(\deg_s(y) = |Y| + 1\). A symmetric argument rules
    out that only \(y\) is active.
\end{proof}

Next, we show that the \texttt{NEAR-OR} gadget permits a PNE in every other
case.

\begin{lemma}\label{near-or-gadget_permissive}
    Let \(G_\lor = (V_\lor, E_\lor)\) an instance of the \texttt{NEAR-OR}
    gadget, \(G = (V, E)\) a graph with \(V \cap V_\lor = X\), and \(H \defas
    (V \cup V_\lor, E \cup E_\lor)\). Then, if \(G\) admits a PNE \(s\) with
    \(\sum_{i=1}^\ell s_{x_i} \not\in \{0, k + 1\}\), then also \(H\) admits a
    PNE \(t\) with \(t_v = s_v\) for all \(v \in V\).
\end{lemma}

\begin{proof}[Proof for \(k = 1\)]
    We claim that \(t\) with
    \begin{align*}
        t_v \defas
        \begin{cases}
            s_v, & \text{if}~v \in V,                          \\
            1,   & \text{if}~v \in \{q\},                      \\
            0,   & \text{if}~v \in
                   V_\lor \setminus \left(X \cup \{q\}\right). \\
        \end{cases}
    \end{align*}
    is a PNE on \(H\). Since \(t_v = s_v\) and, due to \(t_w = 0\), also
    \(\deg_t(v) = \deg_s(v)\) holds for all \(v \in V\) by construction, it
    remains to show that vertices in \((V \cup V_\lor) \setminus V = V_\lor
    \setminus X\) are best-responding. This is the case as \(\deg_t(q) = 0\)
    and \(q\) is active, \(\deg_t(w) = m \not\in \{0, k + 1\}\) and \(w\) is
    inactive, and for all \(v \in \{y, z, y', z'\}\), \(\deg_t(v) = 1\) and
    \(v\) is inactive.
\end{proof}

\begin{proof}[Proof for \(k \geq 2\)]
    Let \(Q \defas Y \cup Z\). We claim that \(t\) with
    \begin{align*}
        t_v \defas
        \begin{cases}
            s_v, & \text{if}~v \in V,                                      \\
            1,   & \text{if}~v \in Q,                                      \\
            0,   & \text{if}~v \in V_\lor \setminus \left(X \cup Q\right). \\
        \end{cases}
    \end{align*}
    is a PNE on \(H\). Again, we only need to show that vertices in \(V_\lor
    \setminus X\) are best-responding. This is the case as for all \(q \in Q\),
    \(\deg_t(q) = 0\) and \(q\) is active, \(\deg_t(w) = m \not\in \{0, k +
    1\}\) and \(w\) is inactive, and \(\deg_t(y) = \deg_t(z) = k\) and \(y\)
    and \(z\) are both inactive.
\end{proof}

Additionally, we argue that the \texttt{NEAR-OR} gadget is safe.

\begin{lemma}\label{near-or-gadget_safe}
    Let \(G_\lor = (V_\lor, E_\lor)\) an instance of the \texttt{NEAR-OR}
    gadget, \(G = (V, E)\) a graph with \(V \cap V_\lor = X\), and \(H \defas
    (V \cup V_\lor, E \cup E_\lor)\). Then, \(H\) admits no PNE in which \(w\),
    the unique vertex in \(N_{G_\lor}(X)\), is active.
\end{lemma}

\begin{proof}[Proof for \(k = 1\)]
    Let \(t\) be a PNE of \(H\) and assume towards a contradiction that \(t_w =
    1\). By \cref{near-or-gadget_restrictive-0,near-or-gadget_restrictive-kp1},
    \(\sum_{i=1}^\ell t_{x_i} \asdef m \not\in \{0, k + 1\}\). Since \(w\) is
    active and \(\deg_t(w) \geq m > 0\), it is \(\deg_t(w) = k + 1 = 2\).
    Therefor, at least one of \(y\) and \(z\) must be active, as otherwise
    \(\deg_t(w) = m\). If both \(y\) and \(z\) are active, then neither \(y'\)
    nor \(q\) can be active as otherwise \(\deg_s(y) > 2\). An analogous
    argument rules out that \(z'\) is active. This implies \(\deg(q) = 2\),
    contradicting \(q\) inactive. Thus, exactly one of \(y\) and \(z\) is
    active, say \(y\). If further \(q\) is active, also \(y'\) is active as
    \(\deg_s(y') = 2\) but this contradicts \(y\) active as then \(\deg_s(y) =
    3\). So \(q\) is inactive, implying that \(z'\) is active due to
    \(\deg_s(z') = 0\). If further \(y'\) is inactive, then \(\deg_s(q) = 2\)
    contradicts \(q\) inactive, so also \(y'\) is active, contradicting
    \(\deg_s(y') = 1\).
\end{proof}

\begin{proof}[Proof for \(k \geq 2\)]
    Let \(t\) be a PNE of \(H\) and assume towards a contradiction that
    \(t_w = 1\). By analogy with the proof for \(k = 1\), we have \(\deg_t(w)
    = k + 1 > m\) so that at least one of \(y\) and \(z\) must be active. If
    \(y\) is active, then \(\deg_t(y') = 1\) for all \(y' \in Y\), so all
    vertices in \(Y\) are inactive and \(\deg_t(y) \in \{1, 2\}\). Since \(2 <
    k + 1\), this contradicts \(y\) active. An analogous argument rules out
    that \(z\) is active.
\end{proof}

Finally, we summarize the behavior of the \texttt{NEAR-OR} gadget.

\begin{corollary}\label{near-or-gadget}
    Let \(G_\lor = (V_\lor, E_\lor)\) an instance of the \texttt{NEAR-OR}
    gadget, \(G = (V, E)\) a graph with \(V \cap V_\lor = X\), and \(H \defas
    (V \cup V_\lor, E \cup E_\lor)\). Then:
    \begin{enumerate}
        \item \texttt{NEAR-OR} is permissive: If \(G\) admits a PNE \(s\) with
              \(\sum_{x \in X} s_x \not\in \{0, k + 1\}\), then also \(H\)
              admits a PNE \(t\) with \(t_v = s_v\) for all \(v \in V\).
        \item \texttt{NEAR-OR} is restrictive: If \(H\) admits a PNE \(t\),
              then \(\sum_{x \in X} t_x \not\in \{0, k + 1\}\).
        \item \texttt{NEAR-OR} is safe: In any PNE \(t\) on \(H\) and for all
              \(m \in N_{G_\lor}(x)\), \(t_m = 0\).
    \end{enumerate}
\end{corollary}

\begin{proof}
    Permissiveness was shown in \cref{near-or-gadget_permissive},
    restrictiveness is the sum of
    \cref{near-or-gadget_restrictive-0,near-or-gadget_restrictive-kp1}, and
    safety follows from \cref{near-or-gadget_safe}.
\end{proof}

\subsection{The \texttt{FALSE} gadget}

\begin{figure}
    \centering
    \begin{subfigure}[b]{0.4\textwidth}
        \centering
        \begin{tikzpicture}
    \input{tikz_header}

    \node[public, label=above:{$x$}] (x) {};
    \node[private, right=4mm of x, label=above:{$y_1$}] (y_1) {};
    \node[private, right=8mm of y_1, label=above:{$y_k$}] (y_k) {};

    \node[hidden, below=8mm of y_1] (_y_1) {}; \node[box, fit=(y_1) (_y_1),
        label={[rotate=-90, xshift=1.5mm]center:\texttt{TRUE}}] {};

    \node[hidden, below=8mm of y_k] (_y_k) {}; \node[box, fit=(y_k)
        (_y_k), label={[rotate=-90, xshift=1.5mm]center:\texttt{TRUE}}] {};

    \node[hidden, above=6mm of x] (_nw) {};
    \node[hidden, right=0mm of y_k] (_se) {};
    \node[box, fit=(_nw) (_se), label={[yshift=3.5mm]center:\texttt{NEAR-OR}}] {};

    \draw[leftout] (y_1) -- (y_k);
\end{tikzpicture}
        \caption{%
            \texttt{FALSE} gadget: In any PNE \(s\) on a graph defined as for
            \cref{fig:near-or-gadgets}, \(s_x = 0\).
        }%
        \label{fig:false-gadget}
    \end{subfigure}
    \hspace{5mm}
    \begin{subfigure}[b]{0.4\textwidth}
        \centering
        \begin{tikzpicture}
    \input{tikz_header}

    \begin{scope}[node distance=4mm and 8mm]
        \node[membrane, label=above:{$y$}] (y) {};
        \node[public, above left=of y, label=left:{$x_1$}] (x_1) {};
        \node[public, below left=of y, label=left:{$x_2$}] (x_2) {};
        \node[private, above right=of y, label=above:{$z_1$}] (z_1) {};
        \node[private, below right=of y, label=below:{$z_k$}] (z_k) {};
    \end{scope}

    \node[hidden, below=10mm of y] (_y) {}; \node[box, fit=(y) (_y),
        label={[rotate=-90, xshift=1.5mm]center:\texttt{FALSE}}] {};

    \node[hidden, right=8mm of z_1] (_z_1) {};
    \node[box, fit=(z_1) (_z_1), label={[xshift=1.5mm]center:\texttt{TRUE}}]
    {};

    \node[hidden, right=8mm of z_k] (_z_k) {};
    \node[box, fit=(z_k) (_z_k), label={[xshift=1.5mm]center:\texttt{TRUE}}] {};

    \draw (y) -- (x_1);
    \draw (y) -- (x_2);
    \draw (y) -- (z_1);
    \draw (y) -- (z_k);
    \draw[leftout] (z_1) -- (z_k);
\end{tikzpicture}
        \caption{%
            \texttt{EQUIV} gadget: In any PNE \(s\) on a graph defined as for
            \cref{fig:near-or-gadgets}, \(s_{x_1} = s_{x_2}\).
        }%
        \label{fig:equiv-gadget}
    \end{subfigure}
    \caption{%
        \texttt{FALSE} and \texttt{EQUIV} gadgets. Only the name and operand
        vertices of auxiliary gadgets are shown.}%
    \label{fig:false-and-equiv-gadgets}
    \caption*{(a) A graph gadget composed of \texttt{NEAR-OR} gadgets: An outer
        \((k + 1)\)-ary one and \(k\) inner \texttt{TRUE} gadgets with common
        operand vertices. The remaining outer operand vertex is labeled \(x\).
        (b) A \texttt{FALSE} gadget with four vertices attached: \(x_1\),
        \(x_2\), and two operand vertices of additional \texttt{TRUE} gadgets.}
\end{figure}

In a \texttt{POSITIVE-1IN3-SAT} instance, literals are either non-negated
boolean variables or falsity (\(\bot\)). To represent the latter, we introduce
a gadget that forces a vertex to be inactive in any PNE
(\cref{fig:false-gadget}).\@

\begin{lemma}\label{false-gadget}
    Let \(G_\bot = (V_\bot, E_\bot)\) an instance of the \texttt{FALSE} gadget,
    \(G = (V, E)\) a graph with \(V \cap V_\bot = \{x\}\), and \(H \defas (V
    \cup V_\bot, E \cup E_\bot)\). Then:
    \begin{enumerate}
        \item \texttt{FALSE} is permissive: If \(G\) admits a PNE \(s\) with
              \(s_x = 0\), then also \(H\) admits a PNE \(t\) with \(t_v =
              s_v\) for all \(v \in V\).
        \item \texttt{FALSE} is restrictive: If \(H\) admits a PNE \(t\), then
              \(t_x = 0\).
        \item \texttt{FALSE} is safe: In any PNE \(t\) on \(H\) and for all \(m
              \in N_{G_\bot}(x)\), \(t_m = 0\).
    \end{enumerate}
\end{lemma}

\begin{proof}
    \emph{Permissiveness.} Let \(G\) admit a PNE \(s\) with \(s_x = 0\). Then,
    the partial strategy profile \(p\) with \(p_v = s_v\), for all \(v \in V\),
    and \(p_{y_i} = 1\), for all \(i \in [k]\), can be extended to a PNE for
    \(H\) by the permissiveness of the \texttt{NEAR-OR} gadget.

    \emph{Restrictiveness.} Assume towards a contradiction that \(H\)
    admits a PNE \(t\) with \(t_x = 1\). Since \texttt{TRUE} is restrictive,
    \(t_{y_i} = 1\) for all \(i \in [k]\). This contradicts \texttt{NEAR-OR}
    being restrictive, as \(t_x + \sum_{i=1}^k t_{y_i} = k + 1\).

    \emph{Safety.} Follows from the safety of the \texttt{NEAR-OR} gadget as
    \(x\) is identified with an operand vertex of the \((k+1)\)-ary
    \texttt{NEAR-OR} gadget in \(G_\bot\).
\end{proof}

\subsection{The \texttt{EQUIV} gadget}

Next, we introduce a gadget to identify equal variables in distinct clauses
(\cref{fig:equiv-gadget}).

\begin{lemma}\label{equiv-gadget}
    Let \(G_\leftrightarrow = (V_\leftrightarrow, E_\leftrightarrow)\) an
    instance of the \texttt{EQUIV} gadget, \(G = (V, E)\) a graph with \(V \cap
    V_\leftrightarrow = \{x_1, x_2\}\), and \(H \defas (V \cup
    V_\leftrightarrow, E \cup E_\leftrightarrow)\). Then:
    \begin{enumerate}
        \item \texttt{EQUIV} is permissive: If \(G\) admits a PNE \(s\) with
              \(s_{x_1} = s_{x_2}\), then also \(H\) admits a PNE \(t\) with
              \(t_v = s_v\) for all \(v \in V\).
        \item \texttt{EQUIV} is restrictive: If \(H\) admits a PNE \(t\), then
              \(t_{x_1} = t_{x_2}\).
        \item \texttt{EQUIV} is safe: In any PNE \(t\) on \(H\) and for all \(m
              \in N_{G_\bot}(\{x_1, x_2\})\), \(t_m = 0\).
    \end{enumerate}
\end{lemma}

\begin{proof}
    \emph{Permissiveness.} Let \(G\) admit a PNE \(s\) with \(s_{x_1} =
    s_{x_2}\). We first show that the strategy profile \(p\) with
    \begin{align*}
        p_v \defas
        \begin{cases}
            s_v, & \text{if}~v \in V,                      \\
            0,   & \text{if}~v = y,                        \\
            1,   & \text{if}~v \in \{z_i \mid i \in [k]\}, \\
        \end{cases}
    \end{align*}
    is a PNE in the graph \(G'\) that is obtained by removing all vertices in
    \(V_\leftrightarrow \setminus \big(\{x_1, x_2, y\} \cup \{z_i \mid i \in
    [k]\}\big)\) from \(H\). Since \(p_y = 0\), all \(v \in V\) are
    best-responding in \(p\), so it remains to show that also \(y\) and \(z_i\)
    for \(i \in [k]\) are best-responding. This is the case as for all \(i \in
    [k]\), \(\deg_p(z_i) = 0\) and \(z_i\) is active, while \(\deg_p(y) \in
    \{k, k + 2\}\) with \(k \geq 1\) and \(y\) is inactive. Since \(p\) is a
    PNE on \(G'\), it follows from the permissiveness of the \texttt{TRUE} and
    \texttt{FALSE} gadgets that also \(H\) admits a PNE \(t\) with \(t_v = p_v
    = s_v\) for all \(v \in V\).

    \emph{Restrictiveness.} Assume towards a contradiction that \(H\) admits a
    PNE \(t\) with \(t_{x_1} \neq t_{x_2}\). Without loss of generality let
    \(t_{x_1} = 1\) and \(t_{x_2} = 0\). By the restrictiveness of the
    \texttt{TRUE} and \texttt{FALSE} gadgets, \(t_y = 0\) and \(t_{z_i} = 1\)
    for all \(i \in [k]\). Thus, \(\deg_t(y) = k + 1\), contradicting \(y\)
    inactive.

    \emph{Safety.} It is \(N_{G_\bot}(\{x_1, x_2\}) = \{y\}\) and, in any PNE
    \(t\) on \(H\), \(t_y = 0\) as \texttt{FALSE} is restrictive.
\end{proof}

\subsection{The \texttt{CLAUSE} gadget}

\begin{figure}
    \centering
    \begin{subfigure}[b]{0.4\textwidth}
        \centering
        \begin{tikzpicture}
    \input{tikz_header}

    \begin{scope}[node distance=7mm]
        \node[public, label=left:{$t_1$}] (t_1) {};
        \node[public, below=of t_1, label=left:{$t_2$}] (t_2) {};
        \node[public, below=of t_2, label=left:{$t_3$}] (t_3) {};
    \end{scope}

    \begin{scope}[node distance=3mm and 10mm]
        \node[membrane, below right=of t_1, label=left:{$x_1$}] (x_1) {};
        \node[membrane, below right=of t_2, label=left:{$y_1$}] (y_1) {};
    \end{scope}

    \begin{scope}[node distance=30mm]
        \node[membrane, right=of t_2, label=left:{$z_1$}] (z_1) {};
    \end{scope}

    \begin{scope}[node distance=3mm]
        \node[membrane, right=of x_1, label=right:{$x_3$}] (x_3) {};
        \node[membrane, right=of y_1, label=right:{$y_3$}] (y_3) {};
        \node[membrane, right=of z_1, label=right:{$z_3$}] (z_3) {};
    \end{scope}

    \node[hidden, left=7mm of t_1] (_nw) {};
    \node[box, fit=(_nw) (t_3), label={[rotate=90, yshift=4mm]center:\texttt{NEAR-OR}}] {};

    \draw (t_1) to[out=0, in=135] (x_1) to[out=225, in=0] (t_2);
    \draw (t_1) to[out=0, in=135] (x_3) to[out=225, in=0] (t_2);
    \draw (t_2) to[out=0, in=135] (y_1) to[out=225, in=0] (t_3);
    \draw (t_2) to[out=0, in=135] (y_3) to[out=225, in=0] (t_3);
    \draw (t_1) to[out=5, in=150] (z_1) to[out=210, in=-5] (t_3);
    \draw (t_1) to[out=5, in=150] (z_3) to[out=210, in=-5] (t_3);

    \draw[leftout] (x_1) -- (x_3);
    \draw[leftout] (y_1) -- (y_3);
    \draw[leftout] (z_1) -- (z_3);

    \node[group, fit=(x_1) (x_3)] {};
    \node[group, fit=(y_1) (y_3)] {};
    \node[group, fit=(z_1) (z_3)] {};
\end{tikzpicture}
        \caption{\texttt{CLAUSE} gadget for \(k = 1\).}%
        \label{fig:clause-gadget-1}
    \end{subfigure}
    \begin{subfigure}[b]{0.4\textwidth}
        \centering
        \begin{tikzpicture}
    \input{tikz_header}

    \begin{scope}[node distance=10mm]
        \node[public, label=above:{$t_1$}] (t_1) {};
        \node[public, below left=of t_1, label=left:{$t_2$}] (t_2) {};
        \node[public, below right=of t_1, label=right:{$t_3$}] (t_3) {};
    \end{scope}

    \draw (t_1) -- (t_2) -- (t_3) -- (t_1);
\end{tikzpicture}
        \caption{\texttt{CLAUSE} gadget for \(k \geq 2\).}%
        \label{fig:clause-gadget-k}
    \end{subfigure}
    \caption{%
        \texttt{CLAUSE} gadget: Admits three symmetric PNE with \(s_{t_1} +
        s_{t_2} + s_{t_3} = 1\).
    }%
    \label{fig:clause-gadgets}
    \caption*{(a) A \texttt{NEAR-OR} gadget whose operand vertices \(t_1\),
        \(t_2\), and \(t_3\) are pairwise connected by three parallel paths
        with one inner vertex each. (b) A triangle with the same operand vertex
        labels.}
\end{figure}

Finally, we represent each clause of a \texttt{POSITIVE-1IN3-SAT} instance by a
\texttt{CLAUSE} gadget (\cref{fig:clause-gadgets}). Here, we do not require the
familiar trio of properties. Instead, the gadget is designed to admit three
symmetric PNE, each of which has one distinct vertex from \(\{t_1, t_2, t_3\}\)
in active state. We prove a slightly weaker claim, which is sufficient for the
reduction.

\begin{lemma}\label{clause-gadget}
    Let \(G\) be an instance of the \texttt{CLAUSE} gadget. Then,
    \begin{itemize}
        \item for every \(t \in \{t_1, t_2, t_3\}\), \(G\) admits a PNE \(s\)
              with \(s_t = 1\), and
        \item in any PNE \(s\) on \(G\), \(s_{t_1} + s_{t_2} + s_{t_3} = 1\).
    \end{itemize}
\end{lemma}

\begin{proof}
    The case of \(k \geq 2\) is trivial: Exactly the profiles in which exactly
    one of the three vertices is active are PNE.\@ Let thus \(k = 1\) in the
    following.

    We show the first claim only for \(t = t_2\) as the other cases are
    analogous. We claim that the partial strategy profile \(s\) with
    \begin{align*}
        s_v \defas
        \begin{cases}
            1, & \text{if}~v \in \{t_2, \; z_1, z_2, z_3\},        \\
            0, & \text{if}~v \in
                 \{t_1, t_3, \; x_1, x_2, x_3, \; y_1, y_2, y_3\}, \\
        \end{cases}
    \end{align*}
    can be extended to a PNE on \(G\). We first argue that \(s\) is a PNE on
    the graph obtained by removing all non-operand vertices of the
    \texttt{NEAR-OR} gadget from \(G\). On this graph it is \(\deg_s(t_2) = 0\)
    and \(t_2\) is active while \(\deg_s(t_1) = \deg_s(t_3) = 3\) and both
    \(t_1\) and \(t_3\) are inactive. Further, \(\deg_s(z) = 0\) and \(z\) is
    active for every \(z \in \{z_1, z_2, z_3\}\) while \(\deg_s(v) = 1\) and
    \(v\) is inactive for all \(v \in \{x_1, x_2, x_3, \; y_1, y_2, y_3\}\).
    Since \(t_1 + t_2 + t_3 = 1 \not\in \{0, k + 1\}\), it follows from
    \cref{near-or-gadget_permissive} that \(s\) can be extended to a PNE on
    \(G\).

    Next, let \(s\) be any PNE on \(G\) and assume towards a contradiction that
    \(\ell \defas s_{t_1} + s_{t_2} + s_{t_3} \neq 1\). The cases of \(\ell =
    0\) and \(\ell = 3\) are ruled out by the restrictiveness of the
    \texttt{NEAR-OR} gadget, so it remains to rule out \(\ell = 2\). Without
    loss of generality let \(s_{t_1} = s_{t_2} = 1\) and \(s_{t_3} = 0\); the
    other cases are analogous. Then, \(\deg_s(x_i) = 2\) and \(x_i\) is active
    for all \(i \in [3]\). Thus, \(\deg_s(t_1) \geq 3\), contradicting \(t_1\)
    active.
\end{proof}

\subsection{Reduction}

With our assortment of gadgets, we can prove the main result of this section,
which answers the second open question posed in~\cite{Gilboa22}:

\begin{theorem}\label{picky-pattern}
    The homogeneous binary public goods game equilibrium decision problem on
    undirected graphs with best-response pattern \(T \in 10^+10^*\) is
    \textsf{NP}-complete.
\end{theorem}

\begin{proof}
    Containment in \textsf{NP} is obvious. For \textsf{NP}-hardness, we
    describe a polynomial-time reduction from \texttt{POSITIVE-1IN3-SAT}. In
    the following, we consider the best-response pattern \(T = 10^k10^*\) for a
    fixed \(k \geq 1\). To ease notation, we write \(s(v)\) instead of \(s_v\)
    to denote the strategy of a vertex \(v\) in a strategy profile \(s\).

    Let \(I \defas \left\{ \{l^i_1, l^i_2, l^i_3\} \mid i \in [\ell] \right\}\)
    with \(l^i_j \in X \cup \{\bot\}\) for all \((i, j) \in [\ell] \times
    [3]\), \(\ell \in \mathbb{Z}_{\geq 1}\), and \(X = \{\xi_1, \ldots,
    \xi_m\}\) a set of boolean variables, be an instance of
    \texttt{POSITIVE-1IN3-SAT}. Recall that \(I\) is a \emph{yes}-instance if
    and only if the formula
    \begin{align*}
        \Phi(X) \defas \bigwedge_{i \in [\ell]} \left(
        \left( l^i_1 \lor l^i_2 \lor l^i_3 \right)
        \land \lnot \left( l^i_1 \land l^i_2 \right)
        \land \lnot \left( l^i_2 \land l^i_3 \right)
        \land \lnot \left( l^i_3 \land l^i_1 \right)
        \right)
    \end{align*}
    is satisfiable. We construct an instance \(G = (V, E)\) of the binary
    public goods game that has a PNE if and only if this is the case. Starting
    from the empty graph, we introduce a disjoint \texttt{CLAUSE} gadget \(C^i
    = (V^i, E^i)\) for every \(\{l^i_1, l^i_2, l^i_3\} \in I\), whose vertices
    we relabel with a superscript \(i\). We call the resulting graph \(G'\).
    Note that for some \(i \neq i' \in [\ell]\) and \(j, j' \in [3]\), it may
    be the case that \(l^i_j = l^{i'}_{j'} \in X\) while \(t^i_j \neq
    t^{i'}_{j'}\) are disjoint vertices. For every such quadruple \((i, j, i',
    j')\), we add an \texttt{EQUIV} gadget on fresh non-operand vertices whose
    operand vertices are identified with \(t^i_j\) and \(t^{i'}_{j'}\). We call
    the graph at this point \(G''\). Next, for every \((i, j) \in [\ell] \times
    [3]\) with \(l^i_j = \bot\), we add a \texttt{FALSE} gadget on fresh
    non-operand vertices whose operand vertex we identify with \(t^i_j\). This
    yields the graph \(G\). Clearly, the number of vertices added is polynomial
    in the number of clauses \(\ell\), so that this construction can be
    accomplished in time polynomial in the size of \(I\). In the following we
    show decision equivalence.

    Let first \(I\) be a \emph{yes}-instance. Then, there is a truth assignment
    \(\sigma \colon X \to \{0, 1\}\) satisfying \(\Phi\). We claim that the
    partial strategy assignment \(s^0\) with
    \begin{align*}
        s^0(t^i_j) \defas
        \begin{cases}
            \sigma(\xi), & \text{if}~l^i_j = \xi \in X, \\
            0,           & \text{if}~l^i_j = \bot,      \\
        \end{cases}
    \end{align*}
    for all \((i, j) \in [\ell] \times [3]\) can be extended to a PNE on \(G\).
    To this end consider first the \texttt{CLAUSE} gadgets and the associated
    subgraph \(G'\) of \(G\). Since \(\Phi\) is satisfied, we have \(s^0(t^i_1)
    + s^0(t^i_2) + s^0(t^i_3) = 1\) for every \(i \in [\ell]\). By
    \cref{clause-gadget}, it follows that \(G'\) admits a PNE \(s'\) with
    \(s'(t^i_j) = s^0(t^i_j)\) for all \((i, j) \in [\ell] \times [3]\).
    Consider next the \texttt{EQUIV} gadgets and the associated subgraph
    \(G''\) of \(G\). Let \(t^i_j\) and \(t^{i'}_{j'}\) be the operand vertices
    of an \texttt{EQUIV} gadget. Then, \(l^i_j = l^{i'}_{j'} \in X\) by
    construction so that \(s'(t^i_j) = s^0(t^i_j) = s^0(t^{i'}_{j'}) =
    s'(t^{i'}_{j'})\) by definition of \(s'\) and \(s^0\). From the
    permissiveness of the \texttt{EQUIV} gadgets (applied iteratively), it
    follows that \(G''\) admits a PNE \(s''\) with \(s''(t^i_j) = s'(t^i_j) =
    s^0(t^i_j)\) for all \((i, j) \in [\ell] \times [3]\). Consider next the
    \texttt{FALSE} gadgets in \(G\). Let \(t^i_j\) be the operand vertex of
    such a gadget. Then, by construction, \(l^i_j = \bot\) and thus
    \(s''(t^i_j) = s^0(t^i_j) = 0\). By the permissiveness of the
    \texttt{FALSE} gadgets (applied iteratively), we have that \(G\) admits a
    PNE as required.

    Let next \(G\) admit a PNE \(s\). We show that the truth assignment
    \(\sigma \colon X \to \{0, 1\}\) given by \(\sigma(l^i_j) \defas s(t^i_j)\)
    for all \((i, j) \in [\ell] \times [3]\) with \(l^i_j \in X\) satisfies
    \(\Phi\). First, we show that \(\sigma\) is well-defined. To this end
    assume towards a contradiction that there are \((i, j) \neq (i', j') \in
    [\ell] \times [3]\) such that \(l^i_j = l^{i'}_{j'}\) but \(s(t^i_j) \neq
    s(t^{i'}_{j'})\). By construction, there is an \texttt{EQUIV} gadget in
    \(G\) with operand vertices \(t^i_j\) and \(t^{i'}_{j'}\), whose
    restrictiveness contradicts \(s\) being a PNE.\@ Next, we show that
    \(\Phi\) is satisfied. Let \(i \in [\ell]\) index a clause \(\{l^i_1,
    l^i_2, l^i_3\} \in I\) and consider the associated \texttt{CLAUSE} gadget
    \(C^i\) with operand vertices \(V^i = \{t^i_1, t^i_2, t^i_3\}\). Since
    \texttt{EQUIV} and \texttt{FALSE} are safe gadgets, implying \(s(v) = 0\)
    for all \(v \in N_G(V^i)\), it follows that \(s\) limited to \(V^i\) is a
    PNE on \(G\). By \cref{clause-gadget}, this implies \(s(t^i_1) + s(t^i_2) +
    s(t^i_3) = 1\). Without loss of generality, assume that \(s(t^i_1) = 1\)
    and \(s(t^i_2) = s(t^i_3) = 0\). We establish that \(l^i_1 \in X\). To this
    end assume towards a contradiction that \(l^i_1 = \bot\). Then, by
    construction, there is a \texttt{FALSE} gadget in \(G\) whose operand
    vertex is \(t^i_1\). Since \texttt{FALSE} is restrictive, \(s(t^i_1) = 1\)
    contradicts \(s\) being a PNE.\@ Since \(l^i_1 \in X\), we have
    \(\sigma(l^i_1) = s(t^i_1) = 1\), so \(\left( l^i_1 \lor l^i_2 \lor l^i_3
    \right)\) is satisfied by \(\sigma\). It remains to show that both
    \(l^i_2\) and \(l^i_3\) are \emph{false} under \(\sigma\). We do so for
    \(l^i_2\); the argument for \(l^i_3\) is analogous. The case of \(l^i_2 =
    \bot\) is clear, so let \(l^i_2 \in X\). Then, \(\sigma(l^i_2) = s(t^i_2) =
    0\). It follows that also \(\lnot \left( l^i_1 \land l^i_2 \right) \land
    \lnot \left( l^i_2 \land l^i_3 \right) \land \lnot \left( l^i_3 \land l^i_1
    \right)\) and, by extension, \(\Phi\) is satisfied by \(\sigma\).
\end{proof}

In~\cite{Gilboa22} it was shown that any \textsf{NP}-hard pattern starting with
a \(1\) remains \textsf{NP}-hard when it is prefixed with the sequence \(10\).
This lets us identify another natural family of patterns that is
\textsf{NP}-hard to decide, that of all \emph{truncated alternating} sequences:

\begin{corollary}\label{truncated-alternating}
    The homogeneous binary public goods game equilibrium decision problem on
    undirected graphs with best-response pattern \(T \in {(10)}^+10^*\) is
    $\mathsf{NP}$-complete.
\end{corollary}

What makes this family interesting is that its ``limit case'', the alternating
sequence \(T = {(10)}^*\), is polynomial-time decidable already for the more
general case of directed graphs~\cite{Papadimitriou21}.

\section{Conclusions}

We studied equilibria of the binary public goods game from the perspective of
computational complexity. We have resolved two open questions posed by
\citet{Gilboa22} that concern the best-response patterns \(1^k0^*\) for \(k
\geq 3\) and \(10^k10^*\) for \(k \geq 1\). For the former family, we
discovered a connection to congestion games, which guarantees the existence of
equilibria and yields a straightforward polynomial time algorithm to compute
one: any sequence of better responses will converge to an equilibrium after at
most \(O(n^2)\) steps. While this holds already for the inhomogeneous case
where each player may follow a different pattern from this family, the problem
becomes \(\mathsf{PLS}\)-complete when we consider in addition links of varying
strength. For the latter family of patterns, we proved that it is
\(\mathsf{NP}\)-hard to decide whether an equilibrium exists. The special case
of \(1010^*\) together with a result in~\cite{Gilboa22} shows that the family
of truncated alternating patterns, \({(10)}^+10^*\), also induces a hard
problem. This complements nicely a positive result for \({(10)}^*\) given
in~\cite{Papadimitriou21}. In concert with concurrent work by~\citet{Gilboa23},
these results settle the computational complexity of computing equilibria for
the homogeneous binary public goods game on general undirected graphs when
players share a finite best-response pattern.

\section*{Acknowledgements}

We are grateful to three anonymous reviewers for helpful comments and
suggestions. We further thank the organizers and participants of the workshop
\emph{Computational Social Dynamics} (Dagstuhl Seminar 22452) for fruitful
discussions. In particular, we thank Noam Nisan for bringing this problem to
our attention.

\bibliographystyle{abbrvnat}
\bibliography{literature}

\end{document}